\documentclass{article}
\usepackage{ijcai25}

\usepackage{epstopdf}
\usepackage{mathrsfs}
\usepackage{mathtools}
\usepackage{tabularx}
\usepackage{multirow}
\usepackage{times}
\usepackage{epstopdf}
\usepackage{soul}
\usepackage{url}
\usepackage[hidelinks]{hyperref}
\usepackage[utf8]{inputenc}
\usepackage[small]{caption}
\usepackage{graphicx}
\usepackage{graphics}
\usepackage{bm}
\usepackage{amsmath}
\usepackage{amsthm}
\usepackage{amsfonts}
\usepackage{amssymb}
\usepackage{booktabs}
\usepackage{algorithm}
\usepackage{algpseudocode}
\usepackage[switch]{lineno}
\usepackage{color}

\usepackage{xspace}
\usepackage{dsfont}
\usepackage{paralist}
\usepackage{nopageno}

\usepackage{amsmath,nccmath}
\usepackage{amssymb}
\usepackage{xspace}
\usepackage{xcolor}
\usepackage{lineno}  

\newcommand{\A}{\mathcal{A}} 
\newcommand{\B}{\mathcal{B}}
\newcommand{\C}{\mathcal{C}} 
\newcommand{\D}{\mathcal{D}}
\newcommand{\E}{\mathcal{E}} 

\newcommand{\G}{\mathcal{G}}

\newcommand{\M}{\mathcal{M}} 

\renewcommand{\O}{\mathcal{O}}

\newcommand{\R}{\mathcal{R}}
\renewcommand{\S}{\mathcal{S}} 
\newcommand{\T}{\mathcal{T}}

\newcommand{\ls}{{\sf{lst}}}

\newcommand{\authorComment}[3]{{\color{#1}\textbf{[\!\![\!\![\marginpar{\centering{\color{#1}\textbf{#2}}}~#2: #3 --- #2.~]\!\!]\!\!]}}}

\newcommand{\ly}[1]{\authorComment{purple}{LY}{#1}}

\newcommand{\alphabet}{\Sigma}
\newcommand{\emptyword}{\epsilon}
\newcommand{\finwords}{\alphabet^*}
\newcommand{\poswords}{\alphabet^+}
\newcommand{\infwords}{\alphabet^\omega}

\newcommand{\states}{Q}

\newcommand{\trans}{\delta}

\newcommand{\init}{q_0}

\newcommand{\wordletter}[2]{#1{[#2]}}

\newcommand{\run}{\rho}
\newcommand{\langsymb}[0]{\mathcal{L}}
\newcommand{\finlang}[1]{\langsymb_{*}(#1)}
\newcommand{\inflang}[1]{\langsymb(#1)}
\newcommand{\flang}[1]{[#1]}
\newcommand{\size}[1]{|#1|}

\newcommand{\fpaths}{\rm{FPaths}}
\newcommand{\ipaths}{\rm{IPaths}}
\newcommand{\distr}{\mathsf{Distr}}

\newcommand{\ltlf}{\text{LTLf}\xspace}
\newcommand{\ppltl}{\text{PPLTL}\xspace}
\newcommand{\ltl}{\text{LTL}\xspace}
\newcommand{\ltlfplus}{\text{LTLf+}\xspace}
\newcommand{\ppltlplus}{\text{PPLTL+}\xspace}
\newcommand{\dfaBuild}{\textsf{DFA}}

\newcommand{\ltlfU}{\textsf{U}}
\newcommand{\ltlfX}{\textsf{X}}
\newcommand{\ltlfN}{\textsf{N}}
\newcommand{\ltlfNeg}{\neg}
\newcommand{\ltlfG}{\textsf{G}}

\newcommand{\ltlfF}{\textsf{F}}

\newcommand{\ltlftrue}{\textsf{true}}
\newcommand{\ltlffalse}{\textsf{false}}

\newcommand{\ppltlY}{\textsf{Y}}
\newcommand{\ppltlS}{\textsf{S}}

\newcommand{\ppltlFirst}{\textsf{first}}
\newcommand{\nextop}{\textsf{next}}

\newcommand{\acccond}{\alpha}

\newcommand{\psem}{\textsf{Psem}}
\newcommand{\psyn}{\textsf{Psyn}}
\newcommand{\pp}{\mathbb{P}}

\newtheorem{theorem}{Theorem}

\newtheorem{proposition}{Proposition}
\newtheorem{example}{Example}

\newcommand{\lang}[1]{\mathcal{L}({#1})}

\newcommand{\buchiaut}{B\"uchi\xspace}

\newcommand{\acc}{F}
\newcommand{\initState}{\iota}

\renewcommand{\B}{\mathcal{B}}
\renewcommand{\M}{\mathcal{M}}
\renewcommand{\C}{\mathcal{C}}

\newcommand{\ap}{\textsf{AP}}
\newcommand{\setnocond}[1]{\{#1\}}

\newcommand{\prob}{\mathsf{P}}
\newcommand{\act}{\text{Act}}
\newcommand{\lab}{\textit{L}}

\title{Solving MDPs with LTLf+ and PPLTL+ Temporal Objectives}

\author{
Giuseppe De Giacomo$^1$
\and 
Yong Li$^{2}$ \footnote{Corresponding author}
\and
Sven Schewe$^3$
\and
Christoph Weinhuber$^1$
\and 
Pian Yu$^{4 \; *}$
\affiliations
$^1$ Department of Computer Science, University of Oxford, UK \\
$^2$ Key Laboratory of System Software (Chinese Academy of Sciences) and State Key  Laboratory of Computer Science, Institute of Software Chinese Academy of Sciences, PRC \\
$^3$ Department of Computer Science, University of Liverpool, UK \\
$^4$ Department of Computer Science, University College London, UK \\
\emails
giuseppe.degiacomo@cs.ox.ac.uk, liyong@ios.ac.cn, sven.schewe@liverpool.ac.uk, christoph.weinhuber@cs.ox.ac.uk, pian.yu@ucl.ac.uk
}

\iftrue
	\usepackage{todonotes}
	\newcommand{\sven}[1]{\todo[inline,color=teal!10,caption={Sven}]{\textbf{Sven:} #1}}
    \renewcommand{\ly}[1]{\todo[inline,color=orange!10,caption={LY}]{\textbf{LY:} #1}}
 	\newcommand{\giuseppe}[1]{\todo[inline,color=magenta!10,caption={Giuseppe}]{\textbf{Giuseppe:} #1}}
\else
	\newcommand{\sven}[1]{}
	\newcommand{\ly}[1]{}
	\newcommand{\pian}[1]{}
 	\newcommand{\giuseppe}[1]{}

\fi

\begin{document}

\maketitle

\setlength{\abovedisplayskip}{1pt}
\setlength{\belowdisplayskip}{1pt}

\begin{abstract}
The temporal logics \ltlfplus and \ppltlplus have recently been proposed to express objectives over infinite traces.
These logics are appealing because they match the expressive power of \ltl on infinite traces while enabling efficient DFA-based techniques, which have been crucial to the scalability of reactive synthesis and adversarial planning in \ltlf and \ppltl over finite traces.
In this paper, we demonstrate that these logics are also highly effective in the context of MDPs.
Introducing a technique tailored for probabilistic systems, we 
leverage the benefits of efficient DFA-based methods and compositionality. This approach is simpler than its non-probabilistic counterparts in reactive synthesis and adversarial planning, as it accommodates a controlled form of nondeterminism (``good for MDPs") in the automata when transitioning from finite to infinite traces. Notably, by exploiting compositionality, our solution is both implementation-friendly and well-suited for straightforward symbolic implementations.

\end{abstract}

\section{Introduction}

Temporal logics are widely used as specification languages in reactive synthesis and adversarial planning~\cite{DBLP:books/daglib/0020348,camacho2019towards}. Among these, linear temporal logic (\ltl)~\cite{pnueli1977temporal} is perhaps the most commonly used. LTL is a formalism used to specify and reason about the temporal behaviour of systems over infinite traces. It has been extensively employed as a specification mechanism for temporally extended goals, as well as for expressing preferences and soft constraints in various fields, including business processes,  robotics, and AI~\cite{bienvenu2011specifying,maggi2011monitoring,fainekos2009temporal}. Linear temporal logic over finite traces (\ltlf) ~\cite{baier2006planning,de2013linear,DBLP:conf/ijcai/GiacomoV15} is a variant of \ltl  with the same syntax but it is interpreted over finite instead of infinite traces.
\ppltl is the pure-past version of \ltlf and scans the trace backwards from the end towards the beginning~\cite{DBLP:conf/ijcai/GiacomoSFR20}. It is well-established that strategy synthesis for \ltlf and \ppltl can be derived from deterministic finite automata (DFA), thereby avoiding the challenges associated with determinising automata for infinite traces, typical of LTL reactive synthesis.

The temporal logics \ltlfplus and \ppltlplus have recently been proposed to express objectives over infinite traces~\cite{DBLP:journals/corr/abs-2411-09366}. These logics are  directly based on Manna and Pnueli's hierarchy of temporal properties \cite{DBLP:conf/podc/MannaP89}. This hierarchy  categorizes temporal properties on infinite traces in \emph{four classes}, which are obtained by requiring that a finite trace property holds for \emph{some prefixes} of infinite traces (``guarantee"), \emph{all prefixes} (``safety"), \emph{infinitely many prefixes} (``recurrence") and for  \emph{all but finitely many prefixes} (``persistence"). Notably every \ltl property can be expressed as a Boolean combination of these four kinds of properties \cite{DBLP:conf/podc/MannaP89}. \ltlfplus and \ppltlplus use respectively \ltlf and \ppltl to express properties over finite traces and obtain the four basic kinds of infinite trace properties of the Manna and Pnueli's hierarchy. This makes them particularly interesting from the computational point of view. While they retain the expressive power of \ltl on infinite traces, they enable the lifting of DFA-based techniques developed for \ltlf and \ppltl to obtain deterministic automata on infinite traces corresponding to formulas, thus avoiding  determinisation of B\"uchi automata, which is known to be a notorious computational bottleneck. As a result, \ltlfplus and \ppltlplus are particularly promising for a number of tasks, such as reactive synthesis \cite{PnueliR89,finkbeiner2016synthesis}, supervisory control for temporal properties \cite{EhlersLTV17}, and planning for 
temporally extended goals in nondeterministic domains \cite{DBLP:journals/amai/BacchusK98,DR-IJCAI18}. 
All these tasks require to obtain from the temporal formula a \emph{deterministic automaton on infinite traces}, to be used as a game arena over which a strategy can be computed to achieve the required property.  
\ltlfplus and \ppltlplus excel at these tasks by enabling simple arena construction through the Cartesian product of DFAs, corresponding to the finite trace (\ltlf/\ppltl) components in the \ltlfplus/\ppltlplus formula.
On the other hand, the game to be solved over this arena is an Emerson-Lei game \cite{EmersonL87}, which requires quite sophisticated techniques \cite{HausmannLP24}.

In this paper, we demonstrate that these logics are even more effective in the context of MDPs.
Traditionally, deterministic Rabin automata have been the standard choice for representing LTL specifications in MDPs~\cite{DBLP:books/daglib/0020348}.
Due to the probabilistic nature of MDPs, the automaton for temporal specifications does not need to be entirely deterministic. State-of-the-art MDP synthesis methods use a restricted form of \buchiaut automata called \emph{Limit-Deterministic Büchi Automata} (LDBAs) for LTL~\cite{DBLP:conf/concur/HahnLST015,DBLP:conf/cav/SickertEJK16,10.24963/ijcai.2023/465}. Recent work has shown that an even more relaxed form of nondeterminism, termed ``\emph{good for MDPs}" (GFM), can be effectively used for MDP synthesis \cite{DBLP:conf/tacas/HahnPSS0W20,DBLP:conf/concur/Schewe0Z23}.

Since this relaxed nondeterminism enables more succinct representations of temporal specifications, we adopt GFM automata in this work. By leveraging GFM's controlled nondeterminism, we present techniques for solving MDPs with \ltlfplus/\ppltlplus objectives that maintain the compositional DFA-based approach of \cite{DBLP:journals/corr/abs-2411-09366}. Instead of using Emerson-Lei automata, our construction obtains simple LDBAs from the DFAs of \ltlf/\ppltl components corresponding to Manna and Pnueli's four classes, then composes them while preserving the ``good for MDPs" property.
The result is a GFM B\"uchi automaton that can be used to solve MDPs in a standard way \cite{DBLP:books/daglib/0020348}. This gives us a simple technique that not only is sound, complete, and computationally optimal, but is both implementation-friendly and well-suited for straightforward symbolic implementation.

\section{Preliminaries}

In the whole paper, we will fix a set of atomic propositions $\ap$.
We denote by $\alphabet = 2^{\ap}$ the set of interpretations over $\ap$; $\alphabet$ is also called the \emph{alphabet} set.
Let $\finwords$ and $\infwords$ denote the set of all finite and infinite sequences, respectively.
The empty sequence is denoted as $\emptyword$ and we index a sequence $u = a_0 a_1 \cdots a_n \cdots$ from $0$.
Moreover, we let $\poswords = \finwords\setminus\setnocond{\emptyword}$.
We denote by $\wordletter{w}{i\cdots j}$ the fragment that starts at position $i$ and ends (inclusively) at position $j$.
Particularly, $\emptyword = \wordletter{w}{i\cdots j}$ for all $w $ if $j < i$.
We denote by $|w|$ the number of letters in $w$ if $w$ is a finite sequence and $\infty$ otherwise.
For a finite or infinite sequence $w$, $\wordletter{w}{0\cdots i}$ is said to be a \emph{prefix} of $w$ if $ 0 \leq i < |w|$.
A \emph{trace} is a \emph{non-empty} finite or infinite sequence of letters in $\alphabet$;
Specially, $\emptyword$ is \emph{not} a trace.

\subsection{\ltlfplus and \ppltlplus over Infinite Traces}
\ltlfplus and \ppltlplus have been derived from \ltlf and \ppltl, respectively~\cite{DBLP:journals/corr/abs-2411-09366}.
The syntax of an \ltlf formula~\cite{baier2006planning,de2013linear} over a finite set of propositions $\ap$ is defined as $\phi ::= a \in \ap \mid \ltlfNeg \phi \mid \phi \land \phi \mid \phi \lor \phi \mid \ltlfX \phi \mid \phi \ltlfU \phi \mid \ltlfF \phi \mid \ltlfG \phi$.
Here $\ltlfX$ (strong Next), $\ltlfU$ (Until), $\ltlfF$ (or $\diamond$) (Finally/Eventually), and $\ltlfG$ (or $\square$) (Globally/Always) are temporal operators. 
The syntax of Pure Past LTL over finite traces (\ppltl) is given as $\phi ::= a \in \ap \mid \neg \phi \mid \phi \land \phi \mid \phi \lor \phi \mid \ppltlY \phi \mid \phi \ppltlS \phi$.
Here $\ppltlY$ (``Yesterday") and $\ppltlS$ (``Since") are the past operators, analogues of ``Next" and ``Until", respectively, but in the past.

Although \ltlf and \ppltl have the same expressive power, translating \ltlf to DFAs requires a doubly exponential blow-up~\cite{de2013linear}, while translating \ppltl to DFAs requires only a singly exponential blow-up~\cite{DBLP:conf/ijcai/GiacomoSFR20}.
We refer interested readers to~\cite{de2013linear} and~\cite{DBLP:conf/ijcai/GiacomoSFR20} for the semantics of \ltlf and \ppltl, respectively.
The language of an \ltlf/\ppltl formula $\phi$, denoted $\flang{\phi}$, is the set of \emph{finite traces} over $2^{\ap}$ that satisfy $\phi$.

The syntax of \ltlfplus (resp. \ppltlplus) is given by the following grammar:
\[ \Psi ::= \forall \phi \mid \exists \phi \mid \forall \exists \phi \mid \exists\forall \phi \mid \Psi \lor \Psi \mid \Psi \land \Psi \mid \neg \Psi\]
where $\phi$ are \emph{finite-trace} formulas in \ltlf/\ppltl over $\ap$.

Let $w \in \infwords$ be an infinite trace and $\phi$ an \ltlf/\ppltl formula.
We use $\models_{+}$ for \emph{non-empty} finite traces and $\models$ for infinite traces.
The semantics of an \ltlfplus/\ppltlplus formula is defined by quantifying over the prefixes of infinite traces:
\begin{itemize}\itemsep=0pt
    \item $w \models \forall \phi$ means that, for all $i \geq 0$, $w[0\cdots i] \models_{+} \phi$.
    \item $w \models \exists \phi$ means that there exists an integer $i \geq 0$ such that $w[0 \cdots i] \models_{+} \phi$.
    \item $w \models \forall \exists \phi$ means that, for every $i \geq 0$, there exists an integer $j \geq i$ such that $w[0\cdots j] \models_{+} \phi$.
    \item $w \models \exists \forall \phi$ means that there exists an integer $i \geq 0$ such that, for all integer $j \geq i$, $w[0\cdots j] \models_{+} \phi$.
\end{itemize}
Similarly, we denote by $\flang{\Psi}$ the set of \emph{infinite traces} satisfying the \ltlfplus/\ppltlplus formula $\Psi$.
It has been shown in~\cite{DBLP:journals/corr/abs-2411-09366} that \ltlfplus, \ppltlplus, and \ltl have the same expressive power.

\subsection{Markov Decision Processes}

Following \cite{DBLP:books/daglib/0020348},
a Markov decision process (MDP) $\M$ is a tuple $(S, \act, \prob, s_0, \lab)$ with a finite set of states $S$, a set of actions $\act$, a transition
probability function $\prob : S \times \act \times S \rightarrow [0, 1]$, an initial state $s_0 \in S$ and a labelling function $\lab : S \rightarrow 2^{\ap}$ that labels a state with a set of propositions that hold in that state.
A path $\xi$ of $\M$ is a finite or infinite sequence of alternating states and actions $\xi=s_0 a_0 s_1 a_1\cdots$, ending with a state if finite, such that for all $i 
\geq 0$, $a_i \in \act(s_i)$ and $\prob(s_i, a_i, s_{i+1}) > 0 $.
The sequence $\lab(\xi)=\lab(s_0)\lab(s_1), \cdots$ over $\ap$ is called the \emph{trace} induced by the path $\xi$ over $\M$.

Denote by $\fpaths$ and $\ipaths$ the set of all finite and infinite paths of $\M$, respectively. A strategy $\sigma$ of $\M$ is a function $\sigma: \fpaths \to \distr(\act)$ such that, for each $\xi\in \fpaths$, $\sigma(\xi)\in \distr(\act({\ls}(\xi)))$, where ${\ls}(\xi)$ is the last state of the finite path $\xi$ and $\distr(\act)$ denotes the set of all possible distributions over $\act$.
Let $\Omega_\sigma^{\M}(s)$ denote the subset of (in)finite paths of $\M$ that correspond to strategy $\sigma$ and initial state $s_0$.

A strategy $\sigma$ of $\M$ is able to resolve the nondeterminism of an MDP and induces a Markov chain (MC) $\M^{\sigma} = (S^+, \prob_{\sigma}, \ap, \lab')$ where for $u = s_0\cdots s_n \in S^+$, $\prob_{\sigma}(u, u\cdot s_{n+1}) = \prob(s_n, \sigma(u), s_{n+1})$ and $\lab'(u) = \lab(s_n)$.

A \emph{sub-MDP} of $\M$ is an MDP $\M' = (S', \act', \prob', \lab)$ where $S' \subseteq S, \act' \subseteq \act$ is such that for every $s \in S'$, $\act'(s)\subseteq \act(s)$, and $\prob'$ and $\lab'$ are obtained from $\prob$ and $\lab$, respectively, when restricted to $S'$ and $\act'$.
In particular, $\M'$ is closed under probabilistic transitions, i.e., for all $s \in S'$ and $a \in\act' $ we have that $\prob'(s, a, s') > 0$ implies that $s' \in S'$.
An \emph{end-component} (EC) of an MDP $\M$ is a sub-MDP $\M'$ of $\M$ such that its underlying graph is strongly connected and it has no outgoing transitions.
A maximal end-component (MEC) is an EC $\E = (E, \act', \prob', \lab)$ such that there is no other EC $\E = (E', \act'', \prob'', \lab)$ such that $E \subset E'$.
An MEC $E$ that cannot reach states outside $E$ is called a \emph{leaf} component.

\begin{theorem}[\cite{DBLP:phd/us/Alfaro97,DBLP:books/daglib/0020348}]\label{thm:end-component}
    Once an end-component $E$ of an MDP is entered, there is a strategy that visits every state-action combination in $E$ with probability $1$ and stays in $E$ forever. Moreover, for every strategy the union of the end-components is visited with probability $1$.
    An infinite path of an MC $\M$ almost surely (with probability $1$) will enter a leaf component.
\end{theorem}

\subsection{Automata}

A (nondeterministic) transition system (TS) is a tuple $\T = (\states, \init, \trans)$, where $\states$ is a finite set of states, $\init \in \states$ is the initial state, and $\trans: \states \times \alphabet \rightarrow 2^{\states}$ is a transition function. 
We also lift $\trans$ to sets as 
$\trans (S, a) := \bigcup_{q\in S}\trans(q, a)$.
A \emph{deterministic} TS is such that if, for each $q \in \states$ and $a \in \alphabet$, $\size{\trans(q, a)} \leq 1$.

An automaton $\A$ is defined as a tuple $(\T, \acccond)$, where $\T$ is a TS and $\acccond$ is an acceptance condition.
A \emph{finite run} of $\A$ on a finite word $u$ of length $n \geq 0$ is a sequence of states $\run = q_{0} q_{1} \cdots q_{n} \in \states^{+}$ such that, for every $0 \leq i < n$, $q_{i+1} \in \trans(q_{i}, \wordletter{u}{i})$, where $\wordletter{u}{i}$ indicates the letter of $u$ in position $i$.

For finite words, we consider finite automata with deterministic TS, known as deterministic finite automata (DFA), where $\acccond = \acc \subseteq \states$ is a set of \emph{final} states.
A finite word $u$ is accepted by the DFA $\A$ if its run $q_0 \cdots q_n$ ends in a final state $q_n \in \acc$.
For an infinite word $w$, a \emph{run} of $\A$ on $w$ is an infinite sequence of states $\run = q_{0}q_1q_2\cdots$ such that, for every $i \geq 0$, $q_{i+1} \in \trans(q_{i}, w[i])$.
Let $\inf(\run)$ be the set of states that occur infinitely often in the run $\run$.
We consider the following acceptance conditions for automata on infinite words:
\begin{itemize}
    \item[\buchiaut/co-\buchiaut.] $\acccond = \acc\subseteq \states$ is a set of \emph{accepting} (\emph{rejecting}, resp.) states for \buchiaut (co-\buchiaut, resp.). A run $\run$ satisfies the \buchiaut (co-\buchiaut, resp.) acceptance condition $\acccond$ if $\inf(\run) \cap \acc \neq \emptyset$ ($\inf(\run) \cap F = \emptyset$, resp.).

    \item[Rabin.] $\acccond = \bigcup^k_{i = 1} \setnocond{(B_i, G_i)}$ is such that $B_i \subseteq \states$ and $G_i\subseteq\states$ for all $1\leq i\leq k$. A run $\rho$ satisfies $\acccond$ if there is some $j \in [1,k]$ such that $\inf(\run) \cap G_j \neq \emptyset$ and $\inf(\run)\cap B_i = \emptyset$.
    
\end{itemize}

A run is \emph{accepting} if it satisfies the condition $\acccond$;
A word $w \in \infwords$ is \emph{accepted} by $\A$ if there is an accepting run $\run$ of $\A$ over $w$.
We use three letter acronyms in $\{D, N\} \times \setnocond{F, B, C, R} \times \setnocond{A}$ to denote automata types where the first letter stands for the TS mode, the second for the acceptance type and the third for automaton.
For instance, DBA stands for deterministic \buchiaut automaton.
An NBA is called a \emph{limit deterministic} \buchiaut automaton (LDBA) if its TS becomes deterministic after seeing accepting states in a run.
We assume that all automata are \emph{complete}, i.e., for each state $s\in\states$ and letter $a\in\alphabet$, $|\trans(s,a)|\geq 1$. 

We denote by $\finlang{\A}$ the set of finite words accepted by a DFA $\A$ or the language of $\A$.
Similarly, we denote by $\inflang{\A}$ the \emph{$\omega$-language} recognized by an $\omega$-automaton $\A$, i.e., the set of $\omega$-words accepted by $\A$.

\section{Classic MDP Synthesis Approach}

In non-probabilistic scenarios, such as reactive synthesis and planning, deterministic $\omega$-automata have to be constructed, e.g.~\cite{DBLP:journals/corr/abs-2411-09366}.
In contrast, in the probabilistic setting, there can be a controlled form of nondeterminism called \emph{good for MDPs (GFM)} due to the effect of probabilities~\cite{DBLP:conf/tacas/HahnPSS0W20,DBLP:conf/concur/Schewe0Z23}.  
More precisely, to synthesise a strategy $\sigma$ for an MDP $\M$ that maximises the satisfaction probability of a given temporal objective $\Psi$, we do the following steps:
first, we construct a GFM automaton $\A$ that recognises $[\Psi]$, then build the product MDP $\M \times\A$, and finally synthesise a strategy $\sigma$ on $\M\times \A$ that maximises the probability of reaching accepting MECs.
Our MDP synthesis algorithm with \ltlfplus and \ppltlplus objectives follows the same methodology; our main contribution is a construction of \ltlfplus/\ppltlplus to GFM automata.

To make our presentation more general, we will assume that we are given a (possibly nondeterministic) $\omega$-automaton $\A = (Q, \delta, q_0, \alpha)$ as specification and an MDP $\M = (S, \act, \prob, s_0, \lab)$.
To find an optimal strategy $\sigma$, we define the semantic satisfaction probability of the induced MC $\M^{\sigma}$ for $\lang{\A}$ as 
$\pp_{\M^{\sigma}}(\lang{\A}) = \pp\setnocond{ \xi \in \Omega_\sigma^{\M}(s_0): \lab(\xi) \in \lang{\A}}$.

For an MDP $\M$, we define the maximal semantic satisfaction probability as
$\psem(\M, \A) = \sup_{\sigma} \pp_{\M^{\sigma}}(\lang{\A})$.
Clearly, for two language-equivalent automata $\A$ and $\B$, it holds that $\psem(\M, \A)= \psem(\M,\B)$.

\paragraph*{Product MDP.}
As aforementioned, we find the strategy that obtains $\psem(\M, \A)$ by building $\M  \times \A$, formally defined as $\M \times \A = (S \times \states, \act \times \states, \prob^{\times}, \langle s_0, q_0\rangle, \lab^{\times}, \acccond^{\times})$ augmented with the acceptance condition $\acccond^{\times}$ where
\begin{itemize}\itemsep=0pt
    \item $\prob^{\times}: (S \times \states) \times (\act\times\states) \times (S \times \states) \rightarrow [0,1]$ such that $\prob^{\times}(\langle s, q\rangle, \langle a, q'\rangle, \langle s', q'\rangle) = \prob(s, a, s')$ if $\prob(s, a, s') > 0$ and $q' \in \trans(q, \lab(s))$,
    \item $\lab^{\times}(\langle s, q\rangle) = \lab(s)$ for a state $\langle s, q\rangle \in S \times \states$, and
    \item For \buchiaut/co-\buchiaut, $\acccond^{\times} = \acc^{\times} = \setnocond{\langle s, q\rangle \in S \times \states: q \in \acc}$, while for Rabin, $B^{\times}_i = \setnocond{\langle s,q\rangle: q \in B_i}$ and $G^{\times}_i = \setnocond{\langle s,q\rangle: q \in G_i}$ for $i \in [1,k]$.
\end{itemize}

Intuitively, $\M\times \A$ resolves the nondeterminism in $\A$ by making each successor an explicit action. Then, we can generate traces in $\lang{\A}$ by enforcing the acceptance condition $\acccond$. 
Now we define the maximal \emph{syntactic} satisfaction probability:
\[\psyn(\M, \A) = \sup_{\sigma} \pp \setnocond{\xi \in \Omega^{\M\times\A}_{\sigma}(\langle s_0,q_0\rangle) : \xi \text{ is accepting}} .\]
Clearly, $\psyn(\M, \A) \leq \psem(\M,\A)$ because accepting runs $\xi$ only occur on accepting words.
Moreover, $\psyn(\M, \A) = \psem(\M,\A)$ if $\A$ is deterministic. 

An EC containing an infinite run that visits all states and transitions yet satisfies $\acccond^{\times}$ is said to be accepting.
According to Theorem~\ref{thm:end-component}, all state-action pairs in an EC can be visited with probability $1$.
Since an accepting run of $\M\times\A$ must eventually enter an accepting MEC \cite{DBLP:books/daglib/0020348}, the syntactic satisfaction probability can be formalised as:
\[\psyn(\M, \A) = \sup_{\sigma}\pp_{(\M\times\A)^{\sigma}}(\diamond X)\] 
where $X$ is the set of states of the accepting MECs (AMECs) in $\M\times\A$.
GFM automata are the $\omega$-automata, whose nondeterminism can be correctly resolved by the strategy. Formally, 
An \emph{$\omega$-automaton $\A$ is  GFM} if, for all finite MDPs $\M$, $\psem(\M, \A) = \psyn(\M,\A)$ holds \cite{DBLP:conf/tacas/HahnPSS0W20}.

Typical GFM automata include: i) deterministic automata, ii) good-for-games automata~\cite{DBLP:conf/csl/HenzingerP06} that have a strategy to produce an accepting run for every accepting word, and iii) LDBAs that satisfy certain conditions~\cite{DBLP:conf/concur/HahnLST015,DBLP:conf/cav/SickertEJK16}.

To use GFM automata in synthesis, we describe a game-theoretic approach to decide what automata are GFM.

\paragraph*{AEC-simulation game.} While determining the GFMness of an NBA is PSPACE-hard~\cite{DBLP:conf/concur/Schewe0Z23}, we can use the two-player \emph{accepting end-component simulation} (AEC simulation) game~\cite{DBLP:conf/tacas/HahnPSS0W20} between Spoiler and Duplicator to prove that our constructed automata are GFM.
Specially, given a GFM automaton $\A$ and an automaton $\B$ with $\lang{\B} = \lang{\A}$, if Duplicator wins the AEC-simulation game, $\B$ is also GFM.

In the game, Spoiler places a pebble on the initial state of $\A$, and Duplicator responds by placing a pebble on the initial state of $\B$. The players moves alternately: Spoiler chooses a letter and transition in $\A$, and Duplicator chooses the corresponding transition over the same letter in $\B$. Unlike classic simulation games, Spoiler can, once during the game, make an AEC claim, that she has reached an AEC and provide all transition sequences that will henceforth occur infinitely often in $\A$. Those transitions cannot be updated afterwards.
The game continues with both players producing infinite runs in their respective automata. Duplicator wins if: (1) Spoiler never makes an AEC claim, (2) the run constructed in $\B$ is accepting, (3) the run constructed in $\A$ does not comply with the AEC claim, or (4) the run constructed in $\A$ is not accepting. 
We say that $\B$ AEC-simulates $\A$, if Duplicator wins.

\begin{theorem}[\cite{DBLP:conf/tacas/HahnPSS0W20}]\label{lem:aec-sim}
    If $\A$ is GFM, $\B$ AEC-simulates $\A$ and $\lang{\B} = \lang{\A}$, then $\B$ is also GFM.
\end{theorem}

The intuition behind Theorem~\ref{lem:aec-sim} is that, for any MDP $\M$, by Theorem~\ref{thm:end-component}, an accepting run of $\M \times \A$ eventually enters an AMEC with probability 1, and the transitions infinitely visited in that AMEC are fixed. When Spoiler makes the AEC claim as the run enters the AMEC, Duplicator can select an accepting run in $\M \times \B$ based on the fixed list of finite traces. Hence, $\psyn(\M, \B) \geq \psyn(\M, \A)$. Since $\lang{\A} = \lang{\B}$ and $\A$ is GFM, we have $\psyn(\M, \B) \geq \psyn(\M, \A) = \psem(\M, \A) \geq \psem(\M, \B)$. The key idea is that, once an AMEC is entered, the list of infinitely visited transitions is fixed due to probability, unlike in the usual simulation game.

GFM automata are more succinct than deterministic automata~\cite{DBLP:conf/cav/SickertEJK16,DBLP:journals/corr/abs-2307-11483}.
This then means that we can obtain smaller product MDP $\M\times\A$ with GFM automata and thus a smaller strategy $\sigma$ since $\sigma$ uses the states in $\M\times\A$ as memory.

We now give a useful observation to prove that our constructed automata from \ltlfplus/\ppltlplus are GFM.
It basically says that language equivalent GFM automata can AEC-simulate each other when an MC is given.
Our proof idea is simple: we just use the optimal strategy $\sigma$ of $\M \times \B$ that obtains $\psem(\M, \B)$ for Duplicator to play against Spoiler.
In this way, Duplicator can always win the AEC-simulation game because $\psem(\M, \A) = \psem(\M, \B)$.

\begin{theorem}\label{lem:bisim}
Let $\A$ and $\B$ be two GFM automata such that $\lang{\A} = \lang{\B}$.
For any MC $\M$, there is a strategy $\sigma$ for $\M \times \B$ to AEC-simulate $\M \times \A$.
\end{theorem}

With these preparations, we present our synthesis approach: Section~\ref{sec:formula-buchi} covers \ltlfplus/\ppltlplus to GFM automata, and Section~\ref{sec:synthesis} outlines our synthesis method.

\section{\ltlfplus/\ppltlplus to GFM B\"uchi Automata}
\label{sec:formula-buchi}

For a given \ltlfplus/\ppltlplus formula $\Psi$, we first construct DFAs for its \ltlf/\ppltl subformulas, followed by automata operations to derive the final GFM automaton.
Existing constructions of GFM LDBAs~\cite{DBLP:conf/concur/HahnLST015,DBLP:conf/cav/SickertEJK16} and GFM NBAs~\cite{DBLP:conf/tacas/HahnPSS0W20} from LTL rely on formula unfolding and construct an explicit state automaton without minimisation.\footnote{For a discussion on minimisation on automata on infinite words, see \cite{DBLP:conf/fsttcs/Schewe10}.} 
In contrast, our approach has two key advantages: 
First, it leverages efficient DFA-based techniques including minimisation at every intermediate step; 
Second, it employs a compositional methodology where the automata for subformulas are constructed and optimised independently before being combined, rather than working with the formula as a whole. Moreover our construction can exploit  symbolic techniques. These aspects allow us to better handle complex specifications by controlling state space growth throughout the process.
This approach applies to \ppltl~\cite{DBLP:conf/ijcai/GiacomoSFR20}.

We assume that we have a construction at hand to efficiently build the DFA for a given \ltlf/\ppltl formula $\phi$.
More precisely, we denote by $\dfaBuild(\phi)$ the DFA constructed for $\phi$.
Note that, $\dfaBuild(\phi)$ does not accept the empty sequence $\emptyword$ by definition\footnote{In some literature, $\emptyword \in [\phi]$ is allowed. In this situation, we must make sure that $\finlang{\dfaBuild(\phi)} = \poswords \cap [\phi]$ since the evaluation of the satisfaction for the \ltlfplus and \ppltlplus formulas quantifies over nonempty finite traces of the infinite trace. }. Figure~\ref{fig:dfa-example} shows a DFA for the \ltlf formula $\phi:=\ltlfF (\textsf{last} \land \text{good})$, alongside a DFA for the language $\flang{\neg\phi}$.

\begin{figure}
    \centering
    \includegraphics[width=0.8\linewidth]{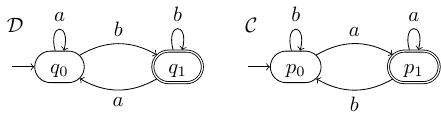}
    \caption{The example DFA $\D$ for \ltlf formula $\phi:=\ltlfF (\textsf{last} \land \text{good})$ that accepts a finite trace in which the proposition $\text{good}$ holds at the last position, and the DFA $\C$ for the language $\flang{\neg\phi}$.
    Here $\textsf{last}$ indicates the last position of a finite trace, i.e.,  $\textsf{last}:= \neg (\ltlfX \ltlftrue)$ and $\alphabet:= 2^{\ap} = \setnocond{a:= \neg \text{good}, b:= \text{good}}$. 
    }
    \label{fig:dfa-example}
\end{figure}

\subsection{Construction for $\exists \phi, \forall \phi, \forall\exists \phi$, and $\exists\forall \phi$}
\label{ssec:constructions}

Next we describe the GFM automata constructions of the formulas $\exists \phi, \forall \phi, \forall\exists \phi$ and $\exists\forall \phi$, which we call \emph{leaf} formulas.

\paragraph{$\exists \phi$.} We first present the construction for $\exists \phi$ below.
\begin{enumerate}\itemsep=0pt
    \item First, let $\D = (\alphabet, Q, \iota, \trans, \acc)$ be $\dfaBuild(\phi)$ such that $\finlang{\D} = [\phi]$.
    \item Second, make all final states in $\acc$ sink final states, i.e., $\trans(s, a) = s$ for each $s \in \acc$ and $a \in \alphabet$, obtaining the automaton $\C = (Q, \iota', \trans', \acc')$.
     $\C$ remains deterministic.
    \item Finally, read it as DBA $\B = ( Q, \iota', \trans', \acc')$.
\end{enumerate}

\begin{theorem}\label{lem:exists-lang}
    $\lang{\B} = \flang{\exists\phi}$.
\end{theorem}
By making final states sink, every word accepted in $\B$ must have a prefix belonging to $[\phi]$.
Hence, the theorem follows.

\paragraph{$\forall \phi$.} Now we introduce the construction for $\forall\phi$.

\begin{figure}[t]
    \centering
    \includegraphics[width=1\linewidth]{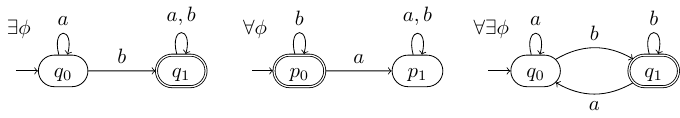}
    \caption{The corresponding \buchiaut automata constructed by our algorithm for $\exists\phi, \forall 
    \phi$ and $\forall \exists \phi$ where $\phi$ is as defined in Figure~\ref{fig:dfa-example}.}
    \label{fig:ba-construction}
\end{figure}

\begin{enumerate}\itemsep=0pt
\item First, let $\C =  (\states, \iota, \trans, \acc)$ be $\dfaBuild(\neg\phi)$ such that $\finlang{\C} = [\neg\phi]$.
\item Then, make all final states in $\C$ as sink final states, remove unreachable states and obtain the DFA $\C' = (Q', \iota, \trans', \acc')$
\item Finally, reverse the set of final states and read it as B\"uchi automaton $\B = (\states', \iota, \trans', \states'\setminus\acc')$.
\end{enumerate}

\begin{theorem}\label{lem:uni-lang}
    $\lang{\B} = \flang{\forall \phi}$.
\end{theorem}
By making final states of $\C$ sink states, every accepting run over a word $w$ in $\B$ does not visit those sink final states in $\C$, which then entails that no prefixes of $w$ belong to $[\neg\phi]$.
That is, all prefixes of $w$ belong to $\phi$. Then the theorem follows.
Our construction for $\forall\phi$ is quite similar to the one 
in~\cite{DBLP:conf/atva/BansalLTVW23}.

\paragraph{$\forall\exists \phi$.}
The construction for $\forall\exists\phi$ is simple and given below.
\begin{enumerate}
    \item First, let $\D = (Q, \iota, \trans, \acc)$ be the DFA $\dfaBuild(\phi)$.
    \item Then, read $\D$ as B\"uchi automaton $\B = (Q, \iota, \trans, \acc)$.
\end{enumerate}

\begin{theorem}\label{lem:uni-exist-phi}
    $\lang{\B} = \flang{\forall \exists \phi}$.
\end{theorem}
Theorem~\ref{lem:uni-exist-phi} clearly holds since every accepting run in $\B$ has a finite prefix run that is an accepting run in $\D$.
Figure~\ref{fig:ba-construction} shows the \buchiaut automata constructed for $\exists\phi, \forall 
    \phi$ and $\forall \exists \phi$.

\paragraph{$\exists\forall \phi$.}
The construction for $\exists\forall \phi$ is more involved and the flowchart is depicted below.
The idea is to first build the DCA $\A$ for $\exists\forall \phi$, which is also the DBA for $\forall\exists \neg \phi$ and then convert $\A$ to the desired LDBA $\B$ accepting $\flang{\exists\forall \phi}$.
In detail:
{\setlength{\intextsep}{5pt} 
\begin{figure}[h]
    \centering
    \resizebox{0.45\textwidth}{!}
    { 
\begin{tikzpicture}[node distance=.5cm]
\tikzstyle{set}=[rounded corners,draw=black!50,thick,inner sep=1pt,minimum width=1.45cm,minimum height=0.225cm]
\tikzstyle{operator}=[circle,draw=black!50,thick]
\begin{scope}
\node [draw, set]  (s1) {\begin{tabular}{c} DFA \\ $\mathcal{D}$ \end{tabular}};
\node [draw, set, right =2cm of s1] (s2) {\begin{tabular}{c} DFA \\ $\mathcal{C}$ \end{tabular}};
\node [draw, set, right =1.2cm of s2]  (s3) {\begin{tabular}{c} DCA \\ $\mathcal{A}$ \end{tabular}};
\node [draw, set, below =.1cm of s2]  (s4) {\begin{tabular}{c} LDBA \\ $\mathcal{B}'$ \end{tabular}};
\node [draw, set, left =2cm of s4]  (s5) {\begin{tabular}{c} LDBA \\ $\mathcal{B}$ \end{tabular}};
\draw[->] (s1) -- (s2) node[midway, above] {complement};
\draw[->] (s2) -- (s3) node[midway, above] {read as};
\draw[->] (s3.south) -- ++(0,0) |- (s4.east) node[pos=0.75, above] {convert};
\draw[->] (s4) -- (s5) node[midway, above] {complete};

\end{scope}

\end{tikzpicture}
} 
\label{fig:existsforall}
\end{figure}
\begin{enumerate}\itemsep=0pt
    \item First, complement $\D$ by reversing the set of final states, and obtain the DFA $\C = (Q, \iota, \trans, \states\setminus\acc)$ for $ \neg \phi$.
    \item Second, read $\C$ as a co-\buchiaut automaton $\A = ( Q, \iota, \trans, \states\setminus\acc)$.
    If we treat $\C$ as a \buchiaut automaton $\G$, by Theorem~\ref{lem:uni-exist-phi}, $\lang{\G} = \lang{\forall\exists \neg\phi}$.
    Since $\A$ is dual to $\G$, we immediately have $\lang{\A} = \infwords \setminus \lang{\G} = \lang{\exists\forall \phi}$.
    \item Third, convert $\A$ to a LDBA $\B' = (\states \times \setnocond{0,1}, \langle \iota, 0\rangle, \trans'', F\times \setnocond{1})$ where $\trans'' = \trans_0 \uplus \trans_j \uplus \trans_1$ is defined as follows:
    \begin{itemize}\itemsep=0pt
        \item $\langle q', 0\rangle = \trans_0(\langle q, 0\rangle, a)$ for all $q \in Q, q' \in Q$ and $a \in \alphabet$ with $\trans(q, a) = q'$,
        \item $\langle q', 1\rangle = \trans_1(\langle q, 1\rangle, a)$ for all $q \in \acc, q' \in \acc$ and $a \in \alphabet$ with $\trans(q, a) = q'$,
        \item $\langle q', 1\rangle = \trans_j(\langle q, 0\rangle, a)$ for all $q \in Q, q' \in F$ and $a \in \alphabet$ with $\trans(q, a) = q'$
    \end{itemize}
    \item Finally, complete the LDBA $\B'$ and obtain the result $\B = (\states', \iota, \trans', \acc' = \acc\times \{1\})$.
\end{enumerate}

An accepting run of $\G$ must visit $\states\setminus \acc$ infinitely often;
equivalently, an accepting run of $\A$ must visit only $\acc$-states from some point on.
The intuition behind the NBA $\B'$ is that it has to guess that point via the transition $\trans_j$.
Thus, before that point, $\B'$ stays within the component $\states\times\setnocond{0}$, and once the run of $\B'$ has entered the component $F \times \setnocond{1}$ via $\trans_j$, it visits only states in $\acc\times \setnocond{1}$ from that moment forward.

\begin{theorem}\label{lem:ex-uni-lang}
    $\lang{\B} = \flang{  \exists  \forall\phi}$.
\end{theorem}

While the DBAs for $\forall\phi$ and $\exists\phi$ can be minimized in polynomial time~\cite{DBLP:journals/ipl/Loding01}, 
minimizing the ones  for $\forall\exists\phi$ and $\exists\forall\phi$ can be NP-complete~\cite{DBLP:conf/fsttcs/Schewe10}. Nonetheless, we can apply cheaper reduction operations such as \emph{extreme minimisation} \cite{DBLP:journals/ijfcs/Badr09} to $\B$.
While the construction gives a LDBA instead of DBAs, the resulting LDBA is also GFM.
\begin{theorem}\label{lem:ex-uni-gfm}
    The automata $\B$ for each of $\exists \phi, \forall \phi, \forall\exists \phi$, and $\exists\forall \phi$ are GFM.
\end{theorem}
\begin{proof}
In fact we have seen that the automaton $\B$ for $\exists \phi, \forall \phi, \forall\exists \phi$ is a DBA so trivially GFM. 
    We prove the result for the automaton $\B$ for $\exists\forall \phi$ by creating an AEC-simulation game between Spoiler and Duplicator.
    Spoiler will play on DCA $\A$ and Duplicator will play on $\B$.
    First, $\lang{\A} = \lang{\B} = [\exists\forall \phi]$.
    Moreover, $\A$ is a GFM automaton because $\A$ is deterministic.
    Therefore, if we can prove that $\B$ AEC-simulates $\A$, then $\B$ is also GFM according to theorem~\ref{lem:aec-sim}.

    Now we provide a winning strategy for Duplicator on $\B$ in the AEC-simulation game.
    Before Spoiler makes the AEC claim, Duplicator will take transitions within $\states \times \setnocond{0}$ via $\trans_0$.
    Once the AEC claim is made by Spoiler, the Duplicator will transition to $F \times \setnocond{1}$ via $\trans_j$ at the earliest point.
    Afterwards, Duplicator takes transitions within $\acc \times \setnocond{1}$ via $\trans_1$.
    The only situation where Spoiler wins is that the run $\rho$ in $\A$ over the chosen word $w$ is accepting.
    That is, $\rho$ cannot visit $Q\setminus F$ states any more after the AEC claim since the list of finite traces visited infinitely often is fixed.
    This then entails that the run $\widehat{\rho}$ constructed by Duplicator will also be accepting in $\B$ since the projection on $Q$ of $\widehat{\rho}$ is exactly $\rho$.
    It follows that Duplicator wins the game and $\B$ AEC-simulates $\A$.
    Hence, $\B$ is also GFM according to Theorem~\ref{lem:aec-sim}.
\end{proof}

\begin{figure}
    \centering
    \includegraphics[width=0.9\linewidth]{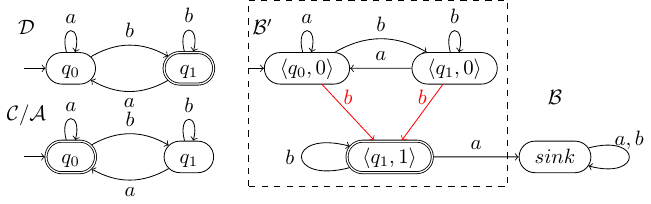}
    \caption{The DFA $\D$ for $\phi$ is on upper left, the DCA $\A$ (or DFA $\C$) with the set of rejecting states $\setnocond{q_0}$ is on lower left, $\A$ has language $\setnocond{a,b}^* \cdot b^{\omega}$ and the LDBA $\B'$ is the part within the dashed box and $\B$ is the complete LDBA where $\setnocond{\langle q_1, 1\rangle}$ is the sole accepting state marked with double rounded boxes and the jump transitions in $\trans_j$ are drawn in red colour.}
    \label{fig:dca-ldba-conversion}
\end{figure}

Following example, illustrates why the constructed LDBA $\B$ for $\exists\forall\phi$ is GFM.

\begin{example}
    Consider the DFA $\D$ for $\phi$ on the upper left of Figure \ref{fig:dca-ldba-conversion}. The DCA $\A$ (or DFA $\C$) with rejecting states $\setnocond{q_0}$ is shown on the lower left, and the complete LDBA $\B$ is on the right (LDBA $\B'$ is within the dashed box). The sole accepting state $\setnocond{\langle q_1, 1\rangle}$ is marked with double rounded boxes, and jump transitions in $\trans_j$ are in red.
    A winning strategy for the Duplicator allows $\B$ to AEC-simulate $\A$. Let $\trans$ and $\trans'$ be the transition functions of $\A$ and $\B$, respectively. The strategy $\sigma$ works as follows:  
    (1) Before the AEC-claim, $\sigma(u) = \trans'(\langle q_0, 0\rangle, u)$;  
    (2) At the AEC claim, $\sigma$ takes the jump transition upon reading $b$ as soon as possible;  
    (3) Afterwards, $\sigma$ uses $\trans'$ for successors.  
    The Spoiler can win only by making an AEC claim and forcing $\A$ to stay in $q_1$ forever. In this case, $\sigma$ ensures an accepting run.
\end{example}

\subsection{Boolean Combinations of GFM Automata}
\label{sec:bool-buchi}

Now, we show that GFM automata are closed under union and intersection.
Let $\A_0$ and $\A_1$ be two GFM \buchiaut automata.
First we introduce the union operation for $\A_0$ and $\A_1$.
\begin{proposition} 
\label{prop:ba-intersection-automaton}
	Given two \buchiaut automata $\A_{0} = (\states_{0}, \initState_{0}, \trans_{0}, \acc_{0})$ and $\A_{1} = (\states_{1}, \initState_{1}, \trans_{1}, \acc_{1})$, let $\A = (\states, \initState, \trans, \acc)$ be the \buchiaut automaton where 
		$\states = \states_{0} \times \states_{1}$, 
		$\initState = (\initState_{0}, \initState_{1})$,
		$\trans(\langle q_{0}, q_{1}\rangle,a) = \trans_0(q_0,a) \times \trans_1(q_1,a)$, and
		$\acc = \states_0 \times \acc_{1} \cup \acc_{0} \times \states_{1}$.
	Then, $\lang{\A} = \lang{\A_{0}} \cup \lang{\A_{1}}$ with $|\states| = |\states_{0}| \cdot |\states_{1}|$.
\end{proposition}

This union operation is just a Cartesian product of $\A_0$ and $\A_1$.
Let $w \in \finwords$.
For the word $w$, a run $\rho_0$ of $\A_0$ and a run $\rho_1$ of $\A_1$ constitute a run $\A$ in the form of $\rho_0 \times \rho_1$.
So, if one of the runs is accepting, $\rho_0\times \rho_1$ is also accepting.
We show below that the union automaton is also GFM.

\begin{theorem}\label{lem:union-gfm}
    If $\A_0$ and $\A_1$ are both GFM, then the union automaton $\A$ is also GFM. 
\end{theorem}
\begin{proof}
Let $\M = (S, \act, \prob, s_0, \lab)$ be an MDP. Our proof goal is to show that $\psyn(\M, \A) = \psem(\M, \A)$.
Then, we can just prove that $\psyn(\M, \A) \geq \psem(\M, \A)$.
We will prove it with the help of equivalent DRAs of $\A_0$ and $\A_1$.

Let $\R_0 = (\states'_0, \trans'_0, \initState'_0, \acccond'_0)$ and $\R_1 = (\states'_1, \trans'_1,\initState'_1, \acccond'_1)$ be two DRAs that are language-equivalent to $\A_0$ and $\A_1$, respectively.
Let $\R=\R_0 \times \R_1$ be the union DRA of $\R_0$ and $\R_1$ such that $\lang{\R} = \lang{\R_0} \cup \lang{\R_1}$.
Formally, $\R$ is a tuple $(\states' = \states'_0\times\states'_1, \trans', \initState'=\langle \initState'_0,\initState'_1\rangle,\acccond')$ where $\trans'(\langle q_0,q_1\rangle,a) = \langle \trans'_0(q_0, a), \trans'_1(q_1, a)\rangle$ for each $\langle q_0,q_1\rangle \in \states'$ and $a \in \alphabet$, and $\acccond' = \bigcup_{i=1}^{k_0}\setnocond{(B_i\times \states'_1, G_i\times\states'_1) : (B_i, G_i)\in \acccond'_0} \cup \bigcup_{i=1}^{k_1} \setnocond{(\states'_0\times B_i, \states'_0\times G_i) : (B_i, G_i)\in \acccond'_1}$.
Let $w \in \infwords$, and $\rho_0$ and $\rho_1$ are the runs over $w$ in $\R_0$ and $\R_1$, respectively.
The run of $\R = \R_0\times\R_1$ over $w$ is actually the product $\rho_0 \times \rho_1$.
Moreover, if $w \in \lang{\A}$, then the run $\rho_0 \times \rho_1$ satisfies either $\acccond'_0$ or $\acccond'_1$, which indicates that $\rho_0 \times \rho_1$ is accepting in $\R $.
Then, it follows that $\psem(\M, \R) = \psem(\M, \A)$ as $\lang{\A} = \lang{\R}$.
As $\R$ is deterministic, we have $\psem(\M, \R) = \psyn(\M, \R)$.
Therefore, we only need to prove that $\psyn(\M, \A) \geq \psyn(\M, \R)$.

Let $\sigma$ be the \emph{optimal} strategy on $\M$ to obtain the maximal satisfaction probability for $\lang{\A}$, i.e., \[\pp_{\M^{\sigma}}(\lang{\A}) = \sup_{\sigma'}\pp\setnocond{ \xi \in \Omega_{\sigma'}^{\M}(s_0): \lab(\xi) \in \lang{\A}}.\]

Note again that here $\sigma$ is usually not a positional strategy for $\M$ and needs extra memory to store history traces.

We denote by $\T_0$, $\T_1$ and $\T_0\times\T_1$ the TSes of $\R_0$, $\R_1$, and $\R$ respectively. 
We now work on the large Markov chain $\M' = \M^\sigma \times \T_0\times \T_1$ by ignoring the acceptance conditions where $\M^{\sigma}$ is already an MC.
Thus, $\M'$ has only probabilistic choices.
Since $\R = \R_0 \times \R_1$ is deterministic, we have that 
\[ \psem(\M, \A) = \psyn(\M, \R) = \psyn(\M', \R).\]

Let $c \in \{0,1\}$.
We know that $\A_c$ is GFM.
According to Theorem~\ref{lem:bisim}, there is an optimal strategy $\sigma_c$ for $\M'\times\A_c$ to AEC-simulate $\M' \times \R_c$.
Thus, we construct the optimal strategy $\sigma^*$ for $\M \times \A$ by building the product of three strategies $\sigma$ for $\M$, $\sigma_0$ and $\sigma_1$, which resolves the nondeterminism of $\M$, $\A_0$ and $\A_1$, respectively, independently in $\M' \times \A$ (viewing $\A$ as the cross product of $\A_0$ and $\A_1$).
That is, in this case, $\sigma_c$ works independently from $\sigma_{1-c}$ on $\M^{\sigma} \times \T_0 \times \T_1 \times \A_c$ and the state space of $\T_0 \times \T_1$ will be used as extra memory for $\sigma_c$ in addition to store states from $\M^{\sigma}$ and $\A_c$.

Then, in the AEC-simulation game, whenever Spoiler produces an accepting run in $\M'\times\R$, we can use $\sigma^*$ to construct an accepting run for Duplicator in $\M'\times \A_0 \times \A_1 = \M' \times \A$.
So, there is a strategy $\sigma^*$ on $\M\times \A$ to achieve $\psyn(\M, \A) = \psyn(\M', \A)\geq \psyn(\M', \R)$.
It then follows that $\psyn(\M, \A) \geq \psyn(\M', \R) = \psem(\M, \A)$.
That is, $\psyn(\M,\A) = \psem(\M,\A)$ for any given MDP $\M$.
Therefore, $\A$ is also GFM.
\end{proof}

We illustrate why the union product $\A$ is GFM if both $\A_0$ and $\A_1$ are GFM with an example.

\begin{example}
    Consider the GFM NBAs $\A_0$ and $\A_1$ on the left of Figure \ref{fig:example-union}, and their union product $\A$ on the right.  
    $\A_0$ accepts $(\setnocond{a,b}^*\cdot b)^{\omega}$, and $\A_1$ accepts $(\setnocond{a,b}^*\cdot a)^{\omega}$. Both are GFM as they have strategies to generate an accepting run for every accepting word.
    The winning strategy stays in the initial state before the AEC-claim and then aims to visit the accepting state repeatedly. The union product $\A$ is also GFM, with a strategy $\sigma$ enabling the Duplicator to win the AEC-simulation game against any deterministic automaton $\D$.  
    The strategy $\sigma$ works as follows:  
    Before the AEC claim, $\sigma$ keeps the run in the initial state $\langle p_0, q_0\rangle$.  
    After the AEC claim, $\sigma$ exits $\langle p_0, q_0\rangle$ and confines the run to $\setnocond{\langle p_0, q_1\rangle, \langle p_1, q_0\rangle}$.  
    For any accepting infinite word chosen by the Spoiler, $\sigma$ ensures an accepting run. In fact, $\sigma$ combines the winning strategies of $\A_0$ and $\A_1$. Hence, $\A$ is GFM.
\end{example}

\begin{figure}
    \centering
    \includegraphics[width=0.9\linewidth]{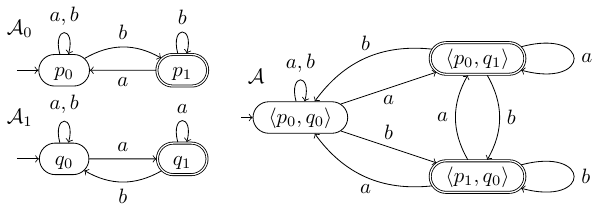}
    \caption{The GFM NBAs $\A_0$ and $\A_1$ are depicted on the left and their union product $\A$ is depicted on the right. The accepting states are marked with double rounded boxes.}
    \label{fig:example-union}
\end{figure}

The Cartesian product is necessary for $\A$ being GFM.
The normal union product only adds an extra initial state to nondeterministically select one of the two automata for the subsequent run, which has only $|\states_1| + |\states_2| + 1$ states.
However, the product might produce an automaton that is not GFM~\cite{DBLP:conf/concur/Schewe0Z23}. 

Next, we introduce the intersection operation for the GFM automata and show that the result automaton is also GFM.
\begin{proposition}[\cite{DBLP:reference/mc/Kupferman18}] 
\label{prop:ba-intersection-automaton}
	Given two \buchiaut automata $\A_{0} = (\states_{0}, \initState_{0}, \trans_{0}, \acc_{0})$ and $\A_{1} = (\states_{1}, \initState_{1}, \trans_{1}, \acc_{1})$, let $\A = (\states, \initState, \trans, \acc)$ be the intersection \buchiaut automaton whose components are defined as follows:  
	\begin{itemize}\itemsep=0pt
    \item $\states = \states_{0} \times \states_{1} \times \setnocond{0,1}$,
    \item $\initState = (\initState_{0}, \initState_{1}, 0)$, 
	\item For a state $\langle q_0, q_1, c\rangle$ and letter $a \in \alphabet$, we have $\langle q'_0,q'_1,\nextop(q_{0}, q_{1}, c)\rangle \in \trans(\langle q_{0}, q_{1}, c\rangle, a)$ where $q'_0 \in \trans_0(q_0, a), q'_1\in\trans_1(q_1,a)$, $\nextop \colon \states_{0} \times \states_{1} \times \setnocond{0,1} \to \setnocond{0,1}$ is defined as 
		\[
			\nextop(q_{0}, q_{1}, c) = 
			\begin{cases}
				1 - c & \text{if $q_{c} \in \acc_{c}$,} \\
				c & \text{otherwise;}
			\end{cases}
		\]
	\item
		$\acc = \acc_{0} \times \states_{1} \times \setnocond{0}$.
	\end{itemize}

	Then, $\lang{\A} = \lang{\A_{0}} \cap \lang{\A_{1}}$ with $|\states| = 2 \cdot |\states_{0}| \cdot |\states_{1}|$.
\end{proposition}

This intersection construction is fairly standard~\cite{DBLP:reference/mc/Kupferman18}. The intuition is to alternatively look for accepting states from $\A_0$ when $c$ is set to $0$ and for accepting states in $\A_1$ when $c$ is $1$.
In this way, the accepting run of $\A$ must visit accepting states from both $\A_0$ and $\A_1$ infinitely often.
The resulting automaton is an LDBA.

Similarly, we show that the GFM automata are also closed under the intersection operation by Theorem~\ref{lem:intersection-gfm}.
    
\begin{theorem}\label{lem:intersection-gfm}
    If $\A_0$ and $\A_1$ are GFM, then the intersection automaton $\A$ is also GFM.
\end{theorem}

In~\cite{DBLP:conf/tacas/HahnPSS0W20}, the constructed GFM automata are \emph{slim} in the sense that each state has at most two successors over a letter. 
The GFM automata we produce can easily be made slim.
This is because the only nondeterministic decision we have to make is to guess in the intersection operation when the individual DCAs will henceforth see only accepting states.
This decision can obviously be arbitrarily delayed.
We can therefore make them round-robin, considering only one LDBA that stems from a DCA at a time. For instance, when we have positive Boolean operations that contain $42$ automata from $\exists\forall\phi$ formulas, then we can avoid the outdegree from jumping to $2^{42}$ at the expense of increasing the state space by a factor of $42$ (while retaining the out-degree of $2$) to handle the round-robin in the standard way.

By applying the constructions of \buchiaut automata for leaf formulas and constructions for intersection and union operations on the syntax tree of a \ltlfplus or \ppltlplus formula, we can obtain a final \buchiaut automaton for the formula.
Therefore, the following result holds.
\begin{theorem}\label{thm:ltl-gfm-buchi}
    For an \ltlfplus/\ppltlplus formula $\Psi$, our method constructs a GFM \buchiaut automaton $\A$ such that $\lang{\A} = [\Psi]$.
\end{theorem}

Let $|\Psi|$ be the length of the formula, i.e., the number of modalities and operations in the formula.
Then we have that:

\begin{theorem}\label{lem:formla-length}
    For an \ltlfplus (respectively, \ppltlplus) formula $\Psi$, the number of states in $\A$ is $2^{2^{\O(|\Psi|)}}$ (respectively $2^{\O(|\Psi|)}$).
\end{theorem}

The number of states obtained in Theorem \ref{lem:formla-length} is  optimal, in the sense that the worst-case double exponential blow-up (single exponential blow-up) is already unavoidable for \ltlf (resp.\ \ppltl) \cite{DBLP:conf/atva/BansalLTVW23,DBLP:conf/ijcai/GiacomoSFR20}.

\subsection{Returning the Strategy}
\label{sec:synthesis}
Given an \ltlfplus/\ppltlplus\ formula $\Psi$ and an MDP $\M$,
our construction can obtain an optimal strategy for $\M$ to achieve maximal satisfaction probability of $\Psi$\, as follows.

First, construct a GFM \buchiaut automaton $\A$ for $\Psi$ using approaches described in Section~\ref{sec:formula-buchi} with $\lang{\A} = \lang{\Psi}$.

Second, construct the product $\M^{\times} = \M \times \A$ and compute the list of AMECs $E = \setnocond{E_1, \cdots, E_k}$ in $\M^{\times}$ using the standard approach described in~\cite{DBLP:books/daglib/0020348}.

Finally, synthesise a strategy as follows: In the AMECs, we can select an action for every state that gives the shortest path to the set of accepting states.
    This shortest path can be multiple but we only need one and the definition of shortest path is clearly well defined.
    For states outside AMECs, we select an action for every state to reach an AMEC with maximal probability;
    this is equivalent to computing the strategy for obtaining the maximal reachability probability to AMECs.
    The resultant strategy, denoted by $\sigma^{\times}$, is positional on $\M^{\times}$ for \buchiaut acceptance condition according to \cite{DBLP:conf/dagstuhl/Mazala01}.

Since the \buchiaut automaton $\A$ for $\Psi$ is GFM we get:
\begin{theorem}\label{thm:main-result}
    The synthesised strategy $\sigma^{\times}$ for $\M$ is  optimal to achieve maximal satisfaction probability of $\Psi$.
\end{theorem}

\section{Conclusion}

In this paper, we investigated the problem of solving MDPs with \ltlfplus and \ppltlplus temporal objectives, showing the effectiveness of these logics for probabilistic planning in MDPs. Our key contribution lies in presenting a provably correct technique to construct GFM Büchi automata for \ltlfplus and \ppltlplus formulas, leveraging the compositional advantages of DFA-based methods. 
Note that our construction is designed to be implementation-friendly and well-suited for a straightforward symbolic implementation. 
In fact, for \ppltl we can directly construct the symbolic DFA in polynomial time  \cite{DBLP:conf/ijcai/GiacomoSFR20}. The final Boolean combination of these symbolic \buchiaut automata, as described in Section~\ref{sec:bool-buchi}, is at most polynomial in their combined sizes and can be realised through a Cartesian product. 
As a future work, we aim to implement this approach, employing symbolic techniques, within state-of-the-art tools such as PRISM \cite{PRISM40}. This development will facilitate the practical application of our methods across a range of domains, including AI, robotics, and probabilistic verification.

\section*{Acknowledgements}
This work has been partially supported by ISCAS Basic Research (Grant Nos.\ ISCAS-JCZD-202406, ISCAS-JCZD-202302), CAS Project for Young Scientists in Basic Research (Grant No.\ YSBR-040), ISCAS New Cultivation Project ISCAS-PYFX-202201,  the ERC Advanced Grant WhiteMech (No.\ 834228), the UKRI Erlangen AI Hub on Mathematical and Computational Foundations of AI, and the EPSRC through grants EP/X03688X/1 and EP/X042596/1.

\bibliographystyle{named}
\bibliography{ref}

\newpage
\cleardoublepage
\appendix

\section*{Appendices}

\subsection*{\ltlf and \ppltl}
Linear Temporal Logic over finite traces (\ltlf)~\cite{baier2006planning,de2013linear} is a variant of \ltl~\cite{pnueli1977temporal} with the same syntax but it is interpreted over finite instead of infinite traces.
The syntax of an \ltlf formula over a finite set of propositions $\ap$ is defined as $\phi ::= a \in \ap \mid \ltlfNeg \phi \mid \phi \land \phi \mid \phi \lor \phi \mid \ltlfX \phi \mid \phi \ltlfU \phi \mid \ltlfF \phi \mid \ltlfG \phi$.
Here $\ltlfX$ (strong Next), $\ltlfU$ (Until), $\ltlfF$ (or $\diamond$) (Finally/Eventually), and $\ltlfG$ (or $\square$) (Globally/Always) are temporal operators interpreted over finite traces.
Note that $\ltlfX$ is a strong next operator such that $\ltlfX \phi$ requires the tail of the finite trace to satisfy $\phi$, while we use $\ltlfN$ to denote the weak next operator such that $\ltlfN \phi$ demands that if the tail of the finite trace is not empty, then it satisfies $\phi$.
Consequently, $\ltlfN \phi := \neg \ltlfX \neg \phi$.
As usual, $\ltlftrue$ and $\ltlffalse$ represent a tautology and a falsum, respectively.
We denote by $\size{\phi}$ the length of $\phi$, i.e., the number of temporal operators and connectives in $\phi$.
We refer interested readers to~\cite{pnueli1977temporal} and~\cite{de2013linear} for the semantics of \ltl and \ltlf, respectively.
The language of an \ltlf formula $\phi$, denoted as $\flang{\phi}$, is the set of \emph{finite traces} over $2^{\ap}$ that satisfy $\phi$.

The syntax of Pure Past LTL over finite traces (\ppltl) is given as $\phi ::= a \in \ap \mid \neg \phi \mid \phi \land \phi \mid \phi \lor \phi \mid \ppltlY \phi \mid \phi \ppltlS \phi$.
Here $\ppltlY$ (``Yesterday") and $\ppltlS$ (``Since") are the past operators, analogues of ``Next" and ``Until", respectively, but in the past.
Note that PPLTL is interpreted over finite traces.
The first position is denoted as $\ppltlFirst$ and we have $\ppltlFirst \equiv \neg (\ppltlY \ltlftrue)$.

Although \ltlf and \ppltl have the same expressive power, it incurs doubly exponential blow-up for translating \ltlf to DFAs~\cite{de2013linear}, while singly exponential blow-up for translating \ppltl to DFAs~\cite{DBLP:conf/ijcai/GiacomoSFR20}.

\subsection*{Detailed AEC game description}

\paragraph{AEC-simulation game.} While determining the GFMness of an NBA is PSPACE-hard~\cite{DBLP:conf/concur/Schewe0Z23}, there is a simple and sufficient way to establish that the \buchiaut automata constructed in this work are GFM by the two-player game \emph{accepting end-component simulation} (AEC simulation)  between \emph{Spoiler} and \emph{Duplicator} ~\cite{DBLP:conf/tacas/HahnPSS0W20}. Specifically, given a GFM automaton $\A$ and an automaton $\B$ such that $\lang{\B} = \lang{\A}$, if Duplicator wins the AEC-simulation game over $\B$ and $\A$, then $\B$ is also GFM.

The AEC-simulation game over $\A$ and $\B$ begins with the Spoiler, who places her pebble on the initial state of $\A$. Then, the Duplicator puts his pebble on the initial state of $\B$.
The two players then take turns and at each round, the Spoiler chooses an input letter and an according transition from $\A$, and then the Duplicator chooses a transition for the same letter in $\B$.
Different from the classic simulation game, in the AEC-simulation game, the
Spoiler has an additional move that she can (and, in order to win, has to) perform once in the game: in addition to choosing a letter and a transition, she can claim that she has reached an AEC, and provide a complete list of sequences of automata transitions that can henceforth occur infinitely often in $\A$.
This list will never be updated afterwards. Both players produce an infinite run of their respective automata. The Duplicator has four ways to win:
\begin{itemize}
    \item if the Spoiler never makes an AEC claim,
    \item if the run of $\B$ he constructs is accepting,
    \item if the run the Spoiler constructs in $\A$ does not comply with the AEC claim, and
    \item if the run that the Spoiler produces is not accepting.
\end{itemize}
We say that $\B$ AEC-simulates $\A$ if the Duplicator wins.

\subsection*{Proof for Theorem~\ref{lem:bisim}}

\begin{proof}
Let $\states_{\A}$ and $\states_{\B}$ be the set of states of $\A$ and $\B$, respectively and let $\lab^{\times}_{\M\times\B}$ be the labelling function of $\M\times\B$.
Since both $\A$ and $\B$ are GFM, we have that $\psyn(\M, \A) = \psem(\M, \A) = \psem(\M, \B) =\psyn(\M, \B)$.
In the AEC-simulation game, Spoiler will play on $\M\times\A$ and the Duplicator will use the \emph{optimal} strategy $\sigma$ to obtain the correct satisfaction probability of $\lang{\B}$ to play on $\M\times\B$.
The only situation where the Spoiler could win the game is that she makes the AEC claim once, the run $\rho \in (S\times\states_{\A})^{\omega}$ constructed by her enters an accepting leaf component and the run is accepting.
This means that all states $(m, q) \in S\times\states_{\A^q}$ in the accepting leaf component have the probability $1$ to create a run with a trace in $\lang{\A^q}$, according to Theorem~\ref{thm:end-component}.
Let $u$ be the finite word that Spoiler selects that lead to the entry state $(m, q)$ in the accepting leaf component.
Assume by contraposition that $\sigma$ is not able to lead $\M\times\B$ over $u$ to a state $(m', q')$ such that from $(m', q')$, the run constructed by $\sigma$ is accepting.
In other words, $(m', q')$ has probability less than $1$ to generate an accepting run in $\M\times\B$.
Then, $\psyn(\M,\B) = \pp_{(\M\times\B)^{\sigma}}\setnocond{\xi \in \Omega^{\M\times\B}_{\sigma}: \xi\text{ is accepting}} < \psem(\M, \B) = \psem(\M,\A)$ since $\pp_{\M\times\B}\setnocond{u w: w \in \infwords} < \pp_{\M\times\A}\setnocond{uw: w\in\infwords} = \pp_{\M}\setnocond{uw: w\in\infwords} $.
This contradicts with the fact that $\B$ is GFM.
Thus, $\sigma$ is able to choose an accepting run of $\M\times \B$ whenever the run of $\M\times\A$ is accepting.
Then the theorem follows.
\end{proof}

\subsection*{Proof for Theorem~\ref{lem:exists-lang}} 
\begin{proof}
\begin{itemize}
    \item[] $w = a_0 a_1\cdots \in \lang{\B}$
    \item[$\Leftrightarrow$] the run $\rho = q_0 q_1  \cdots q^{\omega}_k$ of $\B$ over $w$ such that $q_0  = \iota$ and $q_k \in \acc'$ is the sink accepting state for some $k > 0$.
    \item[$\Leftrightarrow$] $q_0 \cdots q_k$ is an accepting run of $\C$ ($\D$) over $a_0 \cdots a_{k-1}$.
    (The run must be an accepting run of $\D$ as the initial state is not a final state.)
    \item[$\Leftrightarrow$] $a_0\cdots a_{k-1} \in \finlang{\D}$, i.e., $a_0 \cdots a_{k-1} \models_{*} \phi$ 
    \item[$\Leftrightarrow$] $w \models \exists \phi$, i.e., $w \in \flang{\exists\phi}$.
\end{itemize}
\end{proof}

\subsection*{Proof for Theorem~\ref{lem:uni-lang}}

\begin{proof}
 Observe that $\B$ is still deterministic.
\begin{itemize}
    \item[] $w = a_0 a_1\cdots \in \lang{\B}$
    \item[$\Leftrightarrow$] the run $\rho = q_0 q_1  \cdots$ of $\B$ over $w$ such that $q_0  = \iota$ and for all $i \geq  0$, $q_i \in \states'\setminus \acc'$, otherwise it gets trapped in sink states in $\acc'$. 
    \item[$\Leftrightarrow$] for all $k > 0$, every prefix $q_0 \cdots q_k$ is not an accepting run of $\C'$ (equivalently, $\C$).
    \item[$\Leftrightarrow$] for all $k > 0$, $a_0\cdots a_{k-1} \notin \finlang{\C}$, i.e., $a_0 \cdots a_{k-1} \models_{*} \phi$ for all $k > 0$.
    \item[$\Leftrightarrow$] $w \models \forall \phi$, i.e., $w \in \flang{\forall\phi}$.
\end{itemize}    
\end{proof}

\subsection*{Optimisations for the constructions of Section \ref{ssec:constructions}}
For the optimisation of the automata that result from the constructions of Theorems \ref{lem:exists-lang} and \ref{lem:uni-lang}, we note that the leaf formulas in form of $\forall \phi$ and $\exists \phi$ are safety (describing that nothing bad ever happens) and guarantee properties (expressing that good things eventually happen), respectively~\cite{DBLP:conf/podc/MannaP89}, and the resulting DBAs are safety and reachability automata.
They are thus in particular weak.

To explain what weak \buchiaut automata are, we need to first introduce the notion called \emph{strongly connected component} (SCC).
An SCC of a TS $\T=(\states, q_0,\trans)$ is a set of states $C\subseteq \states$ such that for each pair of states $q, q' \in C$, $q$ and $q'$ can reach each other via transitions defined by $\trans$.
A \buchiaut automaton is called \emph{weak} if every SCC contains either all accepting (or: final if red as a DFA) states or all rejecting (or: non-final if red as a DFA) states.
This holds for the automata from our constructions.

Such automata can be \emph{minimised} using the algorithm proposed in~\cite{DBLP:journals/ipl/Loding01}, which itself takes advantage of \emph{DFA minimisation}.

For the construction from Theorem \ref{lem:exists-lang}, this consists of a few simple steps: one would remove all states that are not reachable (states that can only be reached through accepting states in the DFA we start with become non-reachable by the construction),
merge all accepting sinks,
and then successively merge those states to the accepting sink where the accepting sink is reached for every input letter.
Once a fixed point is reached, one can simply minimise the resulting automaton as a DFA.

The construction from Theorem \ref{lem:uni-lang} builds on this construction and can be treated similarly.

For the construction of Theorem \ref{lem:uni-exist-phi}, we obtain a proper (i.e.\ not necessarily weak) DBA, and minimising DBA is hard as shown in \cite{DBLP:conf/fsttcs/Schewe10}.
However, \cite{DBLP:conf/fsttcs/Schewe10} also provides cheap heuristics for a statespace reduction, and one can use this.

The construction of Theorem \ref{lem:ex-uni-lang} results in an LDBA with two deterministic parts.
The second part, where all states are final, can be minimised efficiently.
In order to do so, we can successively remove all states that have no predecessors in $F$ \emph{or} no successor until a fixed point is reached.
We note that a successful AEC claim against the DCA for the relevant part of the proof in Theorem \ref{lem:ex-uni-gfm} cannot contain the counterpart of any of the states pruned this way, as an infinity containing these states would have to contain rejecting states, too, so that the according run in the DCA would be rejecting.

After reaching the fixed point, we can then minimise the resulting automaton using standard DFA minimisation and, where states are merged by this minimisation, re-route the transition from the first part accordingly.

For the first part, we can run an adjusted DFA minimisation where the starting step is to distinguish any two states that have different accepting successors: for states $p,q \in Q \setminus F$ we have $p \not\equiv q$ if there is a letter $a \in \Sigma$ and a state $f \in F$ s.t.\ $\delta(p,a) \ni f \notin \delta(q,a)$.

After this alteration, the quotienting of the standard DFA minimisation is continued as usual.

Note that, while this is a very cheap and simple procedure, it is not a minimisation procedure.

\subsection*{Proof for Theorem~\ref{lem:uni-exist-phi}}

\begin{proof}
    We first observe that $\B$ is deterministic, so that we will refer to \emph{the} run of $\B$ on an input word.
\begin{itemize}
    \item[] $w = a_0 a_1\cdots \in \lang{\B}$
    \item[$\Leftrightarrow$] the run $\rho = q_0 q_1  \cdots$ of $\B$ over $w$ such that $q_0  = \iota$ of $\B$ on $w$ is accepting
    \item[$\Leftrightarrow$] infinitely many prefixes of the run $\rho = q_0 q_1  \cdots$ of $\B$ over $w$ end in an accepting state of $\B$ 
    \item[$\Leftrightarrow$] infinitely many prefixes of $w$ are accepted by $\D$
   \item[$\Leftrightarrow$] infinitely many prefixes of $w$ satisfy $\phi$
   \item[$\Leftrightarrow$] $w \models \forall\exists\phi$, i.e.\ $w \in \flang{\forall\exists\phi}$.
\end{itemize}

This concludes the proof.
\end{proof}

\subsection*{Proof for Theorem~\ref{lem:ex-uni-lang}}

\begin{proof}
    Let $w = a_0 a_1 \cdots$ be an infinite word.
    \begin{itemize}
    \item[] $w = a_0 a_1\cdots \in \lang{\B} = \lang{\B'}$
    \item[$\Leftrightarrow$] the run $\rho = \langle q_0, \ell_0\rangle, \langle q_1, \ell_1\rangle  \cdots \langle q_k, \ell_k\rangle \cdots $ of $\B'$ over $w$ such that $q_0  = \iota$ and for some $k > 0$, we have $\ell_i = 1$ for all $i \geq k$. That is, $\langle q_k, 1\rangle = \trans_j(\langle q_{k-1}, 0\rangle, a_{k-1}) $.
    \item[$\Leftrightarrow$] for some $k > 0$, we have $q_i \in \acc$ for all $i \geq k$.
    \item[$\Leftrightarrow$] for some $k > 0$, $a_0 \cdots a_i \in \finlang{\D}$ or equivalently, $a_0 \cdots a_i \models_* \phi$ for all $i \geq k$.
   \item[$\Leftrightarrow$] $w \models \exists\forall\phi$, i.e.\ $w \in \flang{\exists\forall\phi}$.
\end{itemize}
        This concludes the proof.
\end{proof}

\subsection*{Proof for Theorem~\ref{lem:ex-uni-gfm}}

\begin{proof}
    We will prove the theorem using the AEC-simulation game.
    Recall that a DCA has the same structure as DBAs except that the set $\acc$ is called rejecting sets.
    So, an accepting run of a DCA only visits states outside of $\acc$ from some point on, rather than visits them infinitely often. 
    Therefore, if we read $\C$ as a DCA $\A = (\states, \iota, \trans, \states\setminus \acc)$ as in Step 2, then $\lang{\A} = \flang{\exists\forall \phi} = \lang{\B}$.
    Obviously, $\A$ is GFM since it is deterministic.
    If we can prove that $\B$ AEC-simulates $\A$, then $\B$ is also GFM according to Theorem~\ref{lem:aec-sim}.

    Now we can provide the winning strategy for the Duplicator on $\B$ in the AEC-simulation game.
    Before the Spoiler makes the AEC claim, the Duplicator will take transitions within $\states \times \setnocond{0}$ via the deterministic transition function $\trans_0$.
    $\trans_0$ just mimics the behaviour of the transition function $\trans$ of $\C$.
    Once the AEC claim has been made by Spoiler, the Duplicator will choose to transition to $F \times \setnocond{1}$ via $\trans_j$ at the earliest point.
    Note that $\trans_j$ is also deterministic.
    The only nondeterminism in $\B$ lies in choosing between $\trans_0$ and $\trans_j$ for computing the successors.
    Afterwards, the Duplicator takes transitions within $\acc \times \setnocond{1}$ via the deterministic transition function $\trans_1$.
    The only situation where Spoiler might have a chance to win is that she makes an AEC claim and the run $\rho$ of $\A$ over the chosen word $w$ by her is accepting.
    Since after the AEC claim, the list of finite traces visited infinitely often is fixed, the run $\rho$ must stay within the $\acc$ region; Otherwise $\rho$ will not be accepting because some rejecting states in $\states\setminus\acc$ will be visited also infinitely often.
    
    Assume that $\rho = q_0 \cdots q_{k-1} q_k \cdots \in \states^* \cdot F^{\omega}$ where for some $k > 0$, $q_i \in \acc$ for all $i \geq k$ and Spoiler makes an AEC claim when taking a transition $(q_{\ell}, a_{\ell}, q_{\ell+1})$ for some $\ell > k - 1$. 
    According to our strategy, the run $\widehat{\rho}$ of $\B$ over $w$ created by Duplicator would be $\langle q_0, 0\rangle \cdots \langle q_k, 0\rangle\cdots \langle q_{\ell}, 0\rangle \langle q_{\ell+1}, 1\rangle \cdots$ whose projection on the first component would be exactly the run $\rho$.
    Thus, $\widehat{\rho}$ is accepting in $\B$.
    It follows that $\B$ AEC-simulates $\A$.
    Hence, $\B$ is also GFM according to Theorem~\ref{lem:aec-sim}.
    
\end{proof}

\subsection*{Detailed proof for Theorem~\ref{lem:union-gfm}}
\label{app:lemma-proof--union-gfm}
\begin{proof}
Let $\M = (S, \act, \prob, s_0, \lab)$ be an MDP. Our proof idea is to prove that $\psyn(\M, \A) = \psem(\M, \A)$, which basically requires us to prove that $\psyn(\M, \A) \geq \psem(\M, \A)$.
We will prove it with the help of equivalent DRAs of $\A_0$ and $\A_1$.

Let $\R_0 = (\states'_0, \trans'_0, \initState'_0, \acccond'_0)$ and $\R_1 = (\states'_1, \trans'_1,\initState'_1, \acccond'_1)$ be two DRAs that are language-equivalent to $\A_0$ and $\A_1$, respectively.
Let $\R=\R_0 \times \R_1$ be the union DRA of $\R_0$ and $\R_1$ such that $\lang{\R} = \lang{\R_0} \cup \lang{\R_1}$.
Formally, $\R$ is a tuple $(\states' = \states'_0\times\states'_1, \trans', \initState'=\langle \initState'_0,\initState'_1\rangle,\acccond')$ where $\trans'(\langle q_0,q_1\rangle,a) = \langle \trans'_0(q_0, a), \trans'_1(q_1, a)\rangle$ for each $\langle q_0,q_1\rangle \in \states'$ and $a \in \alphabet$, and $\acccond' = \bigcup_{i=1}^{k_0}\setnocond{(B_i\times \states'_1, G_i\times\states'_1) : (B_i, G_i)\in \acccond'_0} \cup \bigcup_{i=1}^{k_1} \setnocond{(\states'_0\times B_i, \states'_0\times G_i) : (B_i, G_i)\in \acccond'_1}$.
Let $w \in \infwords$, and $\rho_0$ and $\rho_1$ are the runs over $w$ in $\R_0$ and $\R_1$, respectively.
Analogous to intersection product, the run of $\R = \R_0\times\R_1$ over $w$ then is basically the product $\rho_0 \times \rho_1$.
Moreover, if $w \in \lang{\A}$, then the run $\rho_0 \times \rho_1$ satisfies either $\acccond'_0$ or $\acccond'_1$, which indicates that $\rho_0 \times \rho_1$ is accepting in $\R $.
Then, it follows that $\psem(\M, \R) = \psem(\M, \A)$ as $\lang{\A} = \lang{\R}$.
As $\R$ is deterministic, we have $\psem(\M, \R) = \psyn(\M, \R)$.
Therefore, we only need to prove that $\psyn(\M, \A) \geq \psyn(\M, \R)$.

Let $\sigma$ be the \emph{optimal} strategy on $\M$ to obtain the maximal satisfaction probability for $\lang{\A}$, i.e., \[\pp_{\M^{\sigma}}(\lang{\A}) = \sup_{\sigma'}\pp\setnocond{ \xi \in \Omega_{\sigma'}^{\M}(s_0): \lab(\xi) \in \lang{\A}}.\]

Note again that here $\sigma$ is usually not a positional strategy for $\M$ and needs extra memory to store history traces.

We denote by $\T_0$, $\T_1$ and $\T_0\times\T_1$ the TSes of $\R_0$, $\R_1$, and $\R$ respectively. 
Similarly, we now work on the large Markov chain $\M' = \M^\sigma \times \T_0\times \T_1$ by ignoring the acceptance conditions where $\M^{\sigma}$ is already an MC.
Let $S_{\M'}$ be the state space of $\M'$.
Thus, $\M'$ has only probabilistic choices.
Since $\R = \R_0 \times \R_1$ is deterministic, it is easy to see that 
\[ \psem(\M, \A) = \psyn(\M, \R) = \pp_{\M'}(\diamond X)\]
where $X$ is the set of states in accepting MECs.

Let $c \in \{0,1\}$.
We know that $\A_c$ is GFM.
According to Theorem~\ref{lem:bisim}, there is an optimal strategy $\sigma_c$ for $\M'\times\A_c$ to AEC-simulate $\M' \times \R_c$.
Let $\A^q$ be the automaton constructed from $\A$ by setting the initial state to $q$.
Moreover, observe that, for all reachable states $(m, q) \in S_{\M'} \times \states_c$, the syntactic probability of a run starting in $(m,q)$ to create a run with a trace in $\lang{\A^q_c}$ is
\begin{enumerate}
    \item one if $m$ is in a leaf component that satisfies the Rabin condition $\acccond_c$, and
    \item zero if $m$ is in a leaf component that does not satisfy the Rabin condition $\acccond_c$.
\end{enumerate}
Further, for each run from $(m,q)$ and a leaf component $L$ of $\M'$, a state in $L \times Q_c$ is reached in $(\M'\times\A_c)^{\sigma_c}$ with the same probability as $L$ is reached from $m$, as this is simply a projection on the paths of $\M'$.
$\sigma_c$ only resolves the nondeterminism in $\A_c$ and does not impose extra probability.

Thus, we construct the optimal strategy for $\M \times \A$ by building the product of the strategy $\sigma$ for $\M$, $\sigma_0$ and $\sigma_1$, which resolves the nondeterminism of $\A_0$ and $\A_1$, respectively, independently in $\M' \times \A$ (viewing $\A$ as the cross product of $\A_0$ and $\A_1$).
That is, in this case, $\sigma_c$ works independently from $\sigma_{1-c}$ on $\M^{\sigma} \times \T_0 \times \T_1 \times \A_c$ and the state space of $\T_0 \times \T_1$ will be used as extra memory for $\sigma_c$ in addition to store states from $\M^{\sigma}$ and $\A_c$.

It follows that the syntactic probability of a run starting in $(m,\langle q_0,q_1\rangle) \in S_{\M'} \times \states_c$ to create a run with a trace in $\lang{\A^{\langle q_0,q_1\rangle}_c}$ is
\begin{enumerate}
    \item one if $m$ is in a leaf component that satisfies either $\acccond_0$ or $\acccond_1$.
    Since $\M'\times\A_0$ (respectively, $\M'\times\A_1$) AEC-simulates $\M'\times\R_0$ (respectively, $\M'\times\R_1$), this run also visits either $S_{\M'} \times\acc_0$ or $S_{\M'}\times\acc_1$ states infinitely often.
    Hence, the run visits $\M'\times \acc$ of $\M'\times\A$ infinitely often and the run is accepting in $\M' \times \A$.
    \item zero if $m$ is in a leaf component that does satisfy either $\acccond_0$ or $\acccond_1$.
    Similarly, this run must be rejecting in $\M' \times \A$.
\end{enumerate}
It again holds that, for each run from $(m,\langle q_0,q_1\rangle) \in S_{\M'}\times \states$ and a leaf component $L$ of $\M'$, a state in $L \times Q_0 \times Q_1$ is reached in $(\M'\times\A)^{\sigma_0,\sigma_1}$ with the same probability as $L$ is reached from $m$, as this is simply a function of $\M'$.

Therefore, there is a strategy for $\M \times \A$ such that $\psyn(\M, \A) \geq \psyn(\M, \R) = \psem(\M, \A) = \psem(\M, \A)$.
It follows that $\psyn(\M,\A) = \psem(\M,\A)$ for any given MDP $\M$.
This then concludes that $\A$ is also GFM.
\end{proof}

\subsection*{Proof for Theorem~\ref{lem:intersection-gfm}}
The proof of Theorem~\ref{lem:intersection-gfm} is entirely similar to the one of Theorem~\ref{lem:union-gfm}, in which we just replace the Rabin acceptance condition with the Streett acceptance condition for the product automaton.

Since here we use Streett condition, we first introduce everything about it that will be used here.
Similarly to Rabin condition, for Street condition, $\acccond = \bigcup^k_{i = 1} \setnocond{(B_i, G_i)}$ is such that $B_i \subseteq \states$ and $G_i\subseteq\states$ for all $1\leq i\leq k$. 
Recall that for Rabin condition, a run $\rho$ satisfies the acceptance condition if there is some $j \in [1,k]$ such that $\inf(\run) \cap G_j \neq \emptyset$ and $\inf(\run)\cap B_i = \emptyset$, while for Streett condition, a run $\rho$ satisfies the acceptance condition if for \emph{all} $j \in [1,k]$, it holds that $\inf(\run) \cap G_j \neq \emptyset$ or $\inf(\run)\cap B_i = \emptyset$.

In the product MDP $\M\ times \A$ where $\A$ is a Streett automaton, the definition of $\alpha^{\times}$ is the same as Rabin.

\begin{proof}
Let $\M = (S, \act, \prob, s_0, \lab)$ be an MDP. Our proof idea is to prove that $\psyn(\M, \A) = \psem(\M, \A)$, which basically requires us to prove that $\psyn(\M, \A) \geq \psem(\M, \A)$.
We will prove it with the help of equivalent DSAs of $\A_0$ and $\A_1$.

Let $\S_0 = (\states'_0, \trans'_0, \initState'_0, \acccond'_0)$ and $\S_1 = (\states'_1, \trans'_1, \initState'_1, \acccond'_1)$ be two DSAs that are language-equivalent to $\A_0$ and $\A_1$, respectively.
Moreover, let $\S=\S_0 \times \S_1$ be the intersection DSA of $\S_0$ and $\S_1$ such that $\lang{\S} = \lang{\S_0} \cap \lang{\S_1}$.
Formally, we define $\S$ as the tuple $(\states' = \states'_0 \times \states'_1, \trans', \initState' = \langle \initState'_0, \initState'_1\rangle, \acccond')$ where $\trans'(\langle q_0, q_1\rangle, a) = \langle \trans'_0(q_0, a), \trans'_1(q_1, a)\rangle$ for each $\langle q_0, q_1\rangle \in \states'$ and $a \in \alphabet$, and $\acccond' = \bigcup_{i=1}^{k_0}\setnocond{(B_i\times \states'_1, G_i\times\states'_1) : (B_i, G_i)\in \acccond'_0} \cup \bigcup_{i=1}^{k_1} \setnocond{(\states'_0\times B_i, \states'_0\times G_i) : (B_i, G_i)\in \acccond'_1}$.
This construction is fairly standard and we can see that a run $\widehat{\rho}$ of $\S$ is accepting if, and only if, the run of $\widehat{\rho}$ projected down on $\states'_0$ (respectively, $\states'_1$) must satisfies $\acccond'_0$ (respectively, $\acccond'_1$).
Let $w \in \infwords$, and $\rho_0$ and $\rho_1$ are the runs over $w$ in $\S_0$ and $\S_1$, respectively.
In other words, the run of $\S = \S_0\times\S_1$ over $w$ then is basically the product $\rho_0 \times \rho_1$.
Moreover, if $w \in \lang{\A}$, then the run $\rho_0 \times \rho_1$ satisfies $\acccond'$ (i.e., both $\acccond'_0$ and $\acccond'_1$), which indicates that $\rho_0 \times \rho_1$ is accepting in $\S $.


Since $\lang{\A} = \lang{\S}$, we have that $\psem(\M, \S) = \psem(\M, \A)$. 
As $\S$ is deterministic and thus GFM, we have $\psem(\M, \S) = \psyn(\M, \S)$.
Therefore, we only need to prove that $\psyn(\M, \A) \geq \psyn(\M, \S)$.

Let $\sigma$ be the \emph{optimal} strategy on $\M$ to obtain the maximal satisfaction probability for $\lang{\A}$, i.e., \[\pp_{\M^{\sigma}}(\lang{\A}) = \sup_{\sigma'}\pp\setnocond{ \xi \in \Omega_{\sigma'}^{\M}(s_0): \lab(\xi) \in \lang{\A}}.\]

Note that here $\sigma$ is usually not a positional strategy on $\M$ and needs extra memory to store history.

Let $\T_0 = (\states'_0, \trans'_0, \initState'_0)$ and $\T_1 = (\states'_1, \trans'_1, \initState'_1)$ be the TSes of $\S_0$ and $\S_1$, respectively.
Then, the TS of $\S$ is $\T_0 \times \T_1$.
We now work on the large Markov chain $\M' = \M^\sigma \times \T_0 \times \T_1$.
Let $S_{\M'}$ be the state space of $\M'$.
It is easy to see that 
\[ \psyn(\M, \S) = \psem(\M, \S) = \psyn(\M',\S)\]
since $\sigma$ is an optimal strategy for $\M$ to achieve maximal satisfaction of $\lang{\A}$. 
Therefore, to prove that $\psyn(\M, \A) \geq \psyn(\M, \S)$, we only need to prove that $\psyn(\M, \A) \geq \psyn(\M', \S)$. 

Let $c \in \{0,1\}$.
We know that $\A_c$ is GFM.
According to Theorem~\ref{lem:bisim}, there is an optimal strategy $\sigma_c$ for $\M'\times\A_c$ to AEC-simulate $\M' \times \S_c$.
Note that $\sigma_c$ usually needs extra memory to resolve the nondeterminism in $\M'\times \A_c$, since the acceptance condition here is Streett.

For syntactic probabilities, in order to prove that $\psyn(\M, \A) \geq \psyn(\M',\S)$, we will prove that there exists a strategy $\sigma^{*}$ for $\M\times \A$ such that over an infinite path $\rho = s_0 a_0 s_1 a_1\cdots \in S\cdot (\act \times S)^{\omega}$ of $\M$, if there is an accepting run of $\M'\times \S$ containing $\rho$, we can also use $\sigma^*$ to construct an accepting run in $\M\times \A$ that contains $\rho$.
In other words, we can just prove that over an infinite path $\rho = s_0 a_0 s_1 a_1\cdots \in S\cdot (\act \times S)^{\omega}$, if there is an accepting run of the Markov chain $\M'\times \S$ containing $\rho$, the Markov chain $(\M\times \A)^{\sigma^*}$ can also produce an accepting run containing $\rho$. 
Now we construct the optimal strategy $\sigma^*$ for $\M \times \A$ by building the product of the strategy $\sigma$ for $\M$, $\sigma_0$ and $\sigma_1$, which resolve the nondeterminism of $\A_0$ and $\A_1$, respectively, independently in $\M' \times \A$ (viewing $\A$ as the cross product of $\A_0$ and $\A_1$).
In other words, we can define $(\M\times \A)^{\sigma^*}$ as $(\M^{\sigma} \times \T_0\times\T_1 \times \A)^{\sigma_0,\sigma_1} = (\M^{\sigma} \times \T_0\times\T_1 \times \A_0\times \A_1)^{\sigma_0,\sigma_1}$ with the acceptance condition $\acc$.
That is, in this case, $\sigma_c$ works independently from $\sigma_{1-c}$ on $\M^{\sigma} \times \T_0 \times \T_1 \times \A_c$ and the state space of $\T_0 \times \T_1$ will be used as extra memory for $\sigma_c$ in addition to store states from $\M^{\sigma}$ and $\A_c$.

Let $\rho$ be a path of $\M$ generated by the optimal strategy $\sigma$.
Assume that $\M'\times \S = \M^{\sigma} \times \T_0\times\T_1\times \S_0\times \S_1$ is able to produce an accepting run $\widehat{\rho} = \rho \times \rho_0 \times \rho_1 \times \rho_0 \times \rho_1$ where $\rho_0$ and $\rho_1$ are the runs of $\S_0$ and $\S_1$ over $\lab(\rho)$, respectively.
It follows that $\rho_0$ and $\rho_1$ must both be accepting since $\widehat{\rho}$ is accepting in $\M'\times\S $.
Moreover, since $\M'\times \A_0$ AEC-simulates $\M'\times\S_0$ and $\M'\times \A_1$ AEC-simulates $\M\times\S_1$, we can construct an accepting run $\rho\times \rho'_0$ in $\M'\times \A_0$ for $\rho\times \rho_0$ using $\sigma_0$ and an accepting run $\rho\times \rho'_1$ in $\M'\times \A_1$ using $\sigma_1$.
Recall that $\rho$ is constructed by strategy $\sigma$ on $\M$.
Therefore, the strategy $\sigma^{*}$ is able to construct a run in form of $\widehat{\rho}' = \rho \times \rho_0\times \rho_1\times \rho'_0\times \rho'_1\times \setnocond{0,1}^{\omega}$ in $\M'\times \A$.
Since both $\rho'_0$ and $\rho'_1$ are accepting, $\widehat{\rho}'$ must also be accepting by the definition of intersection operation.
Clearly, $\widehat{\rho}'$ has the same probability as $\widehat{\rho}$ and $\rho$.
Therefore, there is a strategy $\sigma^*$ for $\M \times \A$ such that $\psyn(\M, \A) \geq \psyn(\M', \S) = \psem(\M, \A) = \psem(\M, \A)$.
It follows that $\psyn(\M,\A) = \psem(\M,\A)$ for any given MDP $\M$.
This then concludes that $\A$ is also GFM.

\end{proof}

\subsection*{Proof for Theorem~\ref{lem:formla-length}}

\begin{proof}
    Let $n_1$ be the number of conjunctions, $n_2$ be the number of disjunctions and $d_i$ be the length of the $i$-th formula in form of $\forall\psi, \exists\psi,\forall\exists\psi$ and $\exists\forall\psi$, where $1 \leq i \leq k$.
    We call these formulas \emph{leaf formulas}.
    That is, $|\Psi| = n_1 + n_2 + \Sigma_{i=1}^k d_i$.
    Every DFA constructed in the leaf $i$ has $2^{2^{\O(d_i)}}$ states for \ltlfplus formula and $2^{\O(d_i)}$ for \ppltlplus formula.
    Hence, the corresponding \buchiaut automaton in the leaf $i$ has $2\cdot 2^{2^{\O(d_i)}} \in 2^{2^{\O(d_i)}}$ states for \ltlfplus formula and $2\cdot 2^{\O(d_i)} \in 2^{\O(d_i)}$ for \ppltlplus formula.
    We know that the syntax tree of $\Psi$ is a binary tree, the number of internal nodes of $\Psi$ is $n_1 + n_2$, and the number of leaf nodes is $k$.
    The final automaton is basically the result of different cartisian products over the leaf \buchiaut automata.
    After Boolean combination, the resultant automaton will have  $2^{n_1} \cdot \Pi_{i= 1}^k 2^{2^{\O(d_i)}} \leq 2^{|\Psi|} \cdot 2^{\Sigma_{i=1}^k 2^{\O(d_i)}} \leq 2^{|\Psi|} \cdot 2^{ \cdot 2^{\Sigma_{i=1}^k\O(d_i) } } \in 2^{2^{\O(|\Psi|)}}$ (respectively, $2^{n_1} \cdot\Pi_{i= 1}^k 2^{\O(d_i) } \in 2^{\O(|\Psi|)}$ ) for \ltlfplus (respectively, \ppltlplus) formulas.
    The extra $2^{n_1}$ factor is due to the extra bit $c$ for copying the state space in the intersection operations.
    Thus, the theorem follows.
\end{proof}

\end{document}